\documentclass[11pt]{article}

\usepackage[margin=1in]{geometry}

\usepackage{palatino}
\usepackage{graphicx}

\usepackage{enumitem}
\setlist{itemsep=0pt,topsep=2pt}
\setitemize{itemsep=0pt,topsep=2pt}
\setenumerate{itemsep=0pt,topsep=2pt}

\usepackage[
  colorlinks,
  linkcolor = blue,
  citecolor = blue,
  urlcolor = blue]{hyperref}
\usepackage{amsmath,amssymb,amsthm}

\usepackage{mathtools}

\newcommand{\N}{\mathbb{N}}
\newcommand{\C}{\mathbb{C}}

\newcommand{\ket}[1]{|#1\rangle}
\newcommand{\bra}[1]{\langle#1|}
\newcommand{\braket}[2]{\langle#1|#2\rangle}
\newcommand{\proj}[1]{|#1\rangle\langle#1|}
\newcommand{\ketbra}[2]{|#1\rangle\langle#2|}

\DeclarePairedDelimiter{\set}{\lbrace}{\rbrace}

\newcommand{\mc}[1]{\mathcal{#1}}

\newcommand{\ct}{^{\dagger}}
\newcommand{\id}{I}
\newcommand{\ie}{\emph{i.e.}}

\newcommand{\LOCC}{\operatorname{LOCC}}  
\newcommand{\LOCCN}{\ensuremath{\operatorname{LOCC}_{\N}}}  
\newcommand{\ALOCC}{\overline{\operatorname{LOCC}}}  
\newcommand{\bip}[2]{\C^{#1}\otimes\C^{#2}}

\def\a{\alpha}
\def\b{\beta}
\def\t{\theta}
\def\g{\gamma}
\def\<{\langle}
\def\>{\rangle}
\def\C{\mathbb{C}}

\def\be{\begin{equation}}
\def\ee{\end{equation}}
\def\bea{\begin{eqnarray}}
\def\eea{\end{eqnarray}}
\newcommand{\comment}[1]{}

\DeclareMathOperator{\tr}{Tr}

\DeclareMathOperator{\supp}{supp}

\DeclareMathOperator{\rank}{rank}
\DeclareMathOperator{\spec}{spec}
\DeclareMathOperator{\Pos}{Pos}
\DeclareMathOperator{\Herm}{Herm}

\newtheorem{theorem}{Theorem}

\newtheorem{lemma}[theorem]{Lemma}
\newtheorem{cor}[theorem]{Corollary}

\theoremstyle{definition}
\newtheorem{definition}{Definition}

\newcommand{\Thm}[1]{\hyperref[thm:#1]{Theorem~\ref*{thm:#1}}}
\newcommand{\Lem}[1]{\hyperref[lem:#1]{Lemma~\ref*{lem:#1}}}
\newcommand{\Cor}[1]{\hyperref[cor:#1]{Corollary~\ref*{cor:#1}}}
\newcommand{\Def}[1]{\hyperref[def:#1]{Definition~\ref*{def:#1}}}
\newcommand{\Obs}[1]{\hyperref[obs:#1]{Observation~\ref*{obs:#1}}}

\newcommand{\Sect}[1]{\hyperref[sec:#1]{Section~\ref*{sec:#1}}}

\newcommand{\Fig}[1]{\hyperref[fig:#1]{Figure~\ref*{fig:#1}}}
\newcommand{\Tab}[1]{\hyperref[tab:#1]{Table~\ref*{tab:#1}}}
\newcommand{\EqRef}[1]{\hyperref[eq:#1]{(\ref*{eq:#1})}}
\newcommand{\Eq}[1]{Equation~\hyperref[eq:#1]{(\ref*{eq:#1})}}

\begin{document}

\title{When $\ALOCC$ offers no advantage over finite LOCC}

\author{
    Honghao Fu$^{1}$, Debbie Leung$^{1}$, and Laura Man\v{c}inska$^{1,2}$ \\[1.5ex]
    $^{1}$ {\it Institute for Quantum Computing and Department of Combinatorics and Optimization,}\\
    {\it University of Waterloo, Waterloo, Ontario, Canada.} \\ $^{2}$ 
    {\it Centre for Quantum Technologies, National University of Singapore, Singapore.}  
}


\maketitle

\begin{abstract}
We consider bipartite LOCC, the class of operations implementable by
local quantum operations and classical communication between two
parties. Surprisingly, there are operations that cannot be implemented
with finitely many messages but can be approximated to arbitrary
precision with more and more messages. This significantly complicates
the analysis of what can or cannot be approximated with LOCC.  Towards
alleviating this problem, we exhibit two scenarios in which allowing
vanishing error does not help. The first scenario involves
implementation of measurements with projective product measurement
operators.  The second scenario is the discrimination of unextendible
product bases on two $3$-dimensional systems.

\end{abstract}

\section{Introduction}

We consider bipartite finite dimensional quantum systems, and what
state transformations can be achieved given arbitrary quantum
operations on each system and classical communication between them.
This class of quantum operations is known as LOCC.  LOCC arises in
many natural settings.  For example, it is much easier to transmit
classical data than quantum data over long distances.  For another
example, quantum gates involving multiple registers are much harder to
implement, and methods to effect them using entangled states,
measurements, and classical feedback hold high promise.  The study of
LOCC also provides insights on the nature of quantum
information, leading to discoveries including teleportation
\cite{Bennett-1993a}, quantum error correcting codes
\cite{Bennett-1996a}, and security proofs for quantum key exchange
\cite{Lo99b,Shor00}.

Unfortunately, LOCC, as an operationally defined class, does not have
a succinct mathematical description.  Traditionally, one turns to
relaxations of LOCC such as SEP or PPT instead of analysing LOCC
operations. This method has proved fruitful for many problems such as
data hiding \cite{DiVincenzo-2002a} and state discrimination
\cite{Watrous-2005}.  Yet this approach fails to answer other
interesting questions such as whether more rounds of communication, or
infinitely many intermediate measurement outcomes can make a
difference, or whether there are operations that can be approximated
arbitrarily closely with LOCC but do not belong to LOCC (\ie, whether
LOCC is equal to its topological closure $\ALOCC$ or
not). Recent investigation of the LOCC class itself has resolved these
questions; for example, more communication rounds can be helpful
\cite{Chitambar-2011a} and $\LOCC\neq\ALOCC$ 
\cite{Chitambar-2011a,Chitambar-2012a,Chitambar-2012b,Chitambar-2012}.  

A common technique to prove that a certain task cannot be accomplished
perfectly by a finite LOCC protocol is to start by assuming the
contrary. Then the properties of steps taken in any perfect
implementation of the task are shown to be incompatible with the
structure of an LOCC protocol. However, it could still be possible to
accomplish the task with vanishing error using LOCC
protocols. Excluding the possibility of approaching perfect
implementation is much harder and few results have been
established
\cite{IBM,Koashi07,KKB,Childs-2013a,Chitambar-2013a,Chitambar-2013b}.



In this paper, we focus on two problems concerning $\ALOCC$.  Our
first problem is, are there natural classes of measurements in LOCC
that are closed?  In other words, are there sets of measurements such
that it does not help to allow vanishing error and more and more
rounds of communication?  Our second problem concerns the possibility
of discriminating sets of orthogonal product states called UPBs (for
{\em unextendible product bases} though these states are not bases) in
$\ALOCC$.

We first summarize prior works in the two problems of interest.  
The first example of a task for which $\ALOCC$ gives no advantage over
finite LOCC concerns perfect discrimination of the so-called domino
states \cite{IBM}.  They cannot be discriminated by either set.  Over
a decade later a generalization was reached by establishing that the
set of full basis measurements implementable by LOCC is closed
\cite{KKB}. For our second question, \cite{Terhal-1999,DiVincenzo-2003a} 
established that UPBs cannot be distinguished with finite LOCC.  
Reference \cite{DeRinaldis-2004a} studied discrimination of UPBs in 
$\ALOCC$.  However, as pointed out
in \cite{KKB}, the proof in \cite{DeRinaldis-2004a} is incomplete 
and a particular claim in the proof contradicts other proven results.


We make partial progress in the two problems posed above. First, we
show that the set of LOCC implementable projective measurements with
tensor product operators is closed (see \Thm{Measurements}). Second,
we prove that $\ALOCC$ cannot be used to perfectly discriminate states
from a UPB in $\C^3\otimes\C^3$ (see \Thm{upb-no-go}). Both results
are applications of a necessary condition from \cite{KKB}.


The rest of the paper is organized as follows. In \Sect{Prelim} we set
up the notation and give basic definitions. In \Sect{Tensor} we
identify a closed class of projective measurements that can be
implemented with finite LOCC protocols and provide an application of
this result. In \Sect{UPB} we establish that two-qutrit UPBs cannot be
perfectly discriminated even with $\ALOCC$.  We conclude in
\Sect{Conclusion}.

\section{Preliminaries}
\label{sec:Prelim}

\subsection*{Finite and asymptotic LOCC}

We say that a measurement $\mathcal{M}$ can be implemented with \emph{finite LOCC}, and write $\mathcal{M}\in\LOCCN$, if $\mathcal{M}$ can be implemented exactly using a finite LOCC protocol (\ie, an LOCC protocol with finitely many communication rounds). We say that $\mathcal{M}$ can be implemented using \emph{asymptotic LOCC}, and write $\mathcal{M}\in\ALOCC$, if there exists a sequence of finite LOCC protocols $\mathcal{P}_1,\mathcal{P}_2,\dotsc$ that implement $\mathcal{M}$ with vanishing error. Note that asymptotic LOCC is the (topological) closure of finite LOCC. It represents the set of operations that can be implemented by LOCC protocols with arbitrary precision. For more detailed explanation of classes $\LOCCN$ and $\ALOCC$ see \cite{Chitambar-2012,M13}.

\subsection*{State discrimination problem}
Let $[n]:=\set{1,\dotsc,n}$, $\pi:[n] \to [0,1]$ be a probability distribution, and $S=\set{\rho_i : i\in[n]} \subseteq \bip{d_A}{d_B}$ be a set of bipartite states. In the state discrimination problem an index $i$ is chosen with probability $\pi(i)$ and Alice and Bob are given their respective registers of $\rho_i$. Their task is to find the index $i$ without error. The scenario where they are allowed to err is also of interest but will not be studied in this paper. In the error free case the probability distribution $\pi$ is not relevant and will therefore be chosen to be uniform. We say that the states from $S$ can be discriminated with finite LOCC (asymptotic LOCC), if there exists a measurement $\mathcal{M}\in\LOCCN$ ($\mathcal{M}\in\ALOCC$) that discriminates the states perfectly. 

There is a close connection between discrimination of mutually orthogonal states and implementation of projective measurements. In particular, the states from $S$ can be discriminated by finite LOCC (asymptotic LOCC) if and only if finite LOCC (asymptotic LOCC) can be used to implement the projective measurement onto the supports of the states $\rho_i$ \cite{Childs-2013a}.

\subsection*{Non-disturbing measurements}
We now introduce the concept of non-disturbing operators which is central for state discrimination with LOCC.
\begin{definition}
Let $S\subseteq\Pos(\C^d)$ be a set of orthogonal states. We say that $E \in \Pos(\C^d)$ is \emph{non-disturbing for $S$}, if
\begin{equation}
  \tr\big(E \rho E \sigma\big) = 0
\label{eq:Nondisturbing}
\end{equation}
for all distinct $\rho,\sigma\in S$. We say that a measurement $\mathcal{M}$ is \emph{non-disturbing for $S$} if each of its POVM elements of $\mathcal{M}$ is non-disturbing for $S$.
\label{def:Nondisturbing}
\end{definition}

Let $\supp(M)$ denote the support of $M$. Then Condition~(\ref{eq:Nondisturbing}) is equivalent to requiring that for all distinct $\rho,\sigma\in S$ and all $\ket{\psi} \in \supp(\rho)$ and $\ket{\phi}\in \supp(\sigma)$
\begin{equation}
  \bra{\psi} E \ket{\phi} = 0.
\end{equation}
Note that any measurement protocol transforms $S$ to a new set
conditioned on the culmulative measurement outcome.  In a perfect
discrimination protocol for $S$, at any point, the next measurement
applied to this conditioned set must not disturb it.  In particular,
the protocol must start with a measurement that is non-disturbing for
$S$. In an LOCC protocol each measurement must be local.  For finite
LOCC, each measurement has to be non-trivial. Hence, the states from a
set $S$ can be perfectly discriminated with finite LOCC only if $S$
admits a non-disturbing product operator $a\otimes b$ where exactly
one of the matrices $a,b$ is the identity matrix. If such an operator
does not exist, the states from $S$ cannot be discriminated with
finite LOCC.  Non-disturbing operators also provide a necessary
condition for state discrimination with asymptotic LOCC.

\begin{theorem}[\cite{KKB}]
Consider a set of states $S = \set{\rho_1,\dotsc,\rho_n}\subseteq\Pos(\bip{d_A}{d_B})$ such that $\bigcap_i \ker \rho_i$ does not contain any nonzero product vector. 
Then $S$ can be discriminated with asymptotic LOCC only if for all $\chi$ with $1/n\leq \chi \leq 1$ there exists a positive semidefinite product operator $E=a\otimes b$ satisfying all of the following:
\begin{enumerate}
\item $\sum_{\rho\in S}\tr(E\rho)=1$, 
\item $\max_{\rho\in S}\tr(E\rho)=\chi$,
\item $E$ is non-disturbing for $S$.
\end{enumerate}
\label{thm:KKB}
\end{theorem}

\Thm{KKB} implies that the set of LOCC implementable full basis measurements is closed. 
\begin{cor}[\cite{KKB}]
If a full orthogonal basis measurement can be implemented using asymptotic LOCC then it can already be implemented with finite LOCC.
\label{cor:KKB}
\end{cor}
In the next section we generalize \Cor{KKB} for a larger class of projective measurements.

\section{Projective measurements with tensor product operators}
\label{sec:Tensor}
In light of the findings of \cite{KKB} presented in the previous section
it is natural to ask whether a similar result holds for the class of all projective measurements, or even for the class of all POVM measurements. To this end, in this section we show that the class of all projective measurements with tensor product operators that can be implemented with finite LOCC is indeed closed.

\subsection{Results}

\begin{theorem}
Let $\mathcal{M}=\set{P_i^A\otimes P_i^B}_{i\in[n]}\subseteq\Pos(\bip{d_A}{d_B})$ be a projective measurement. Then $\mathcal{M}\in\overline{\LOCC}$ implies that $\mathcal{M}\in\LOCC_{\N}$.
\label{thm:Measurements}
\end{theorem}

We say that a set of states $S= \set{\rho_i}_i\subseteq\Pos(\bip{d_A}{d_B})$ is a \emph{full orthogonal set} if the states $\rho_i$ are mutually orthogonal, and $\sum_{i} \rho_i$ has full rank. We now rephrase \Thm{Measurements} in terms of state discrimination. 

\begin{theorem}
Let $S= \set{\rho_i := \tau_i \otimes \sigma_i}_{i\in[n]}\subseteq\Pos(\bip{d_A}{d_B})$ be a full orthogonal set. If the states from $S$ can be discriminated with asymptotic LOCC then they can be discriminated with finite LOCC.
\label{thm:States}
\end{theorem}

We first prove two lemmas that help establish \Thm{States}.  The main
ingredient for proving \Thm{States} is the construction of an operator
$m$ such that either $m \otimes \id$ or $\id \otimes m$ is
non-disturbing for $S$. A matrix $M$ is non-disturbing for a complete
orthogonal set $S= \set{\rho_i}_i$ if and only if the row spaces
$\mathcal{H}_i$ of $\rho_i$ are $M$-invariant for each $i$. Here, a
subspace $\mathcal{H}$ is said to be $M$-invariant, if $M \ket{\psi}
\in \mathcal{H}$ for all $\ket{\psi}\in\mathcal{H}$. 

%


Our first lemma provides useful characterization of $M$-invariance.


\begin{lemma}
Let $M\in\Herm(\C^d)$, $\mathcal{H}$ be an $h$-dimensional 
subspace of $C^d$, 
$\set{\ket{v_i}}_{i\in[h]}$ be a fixed orthonormal basis of
$\mathcal{H}$, $\ket{i}\in\C^h$, 
and $Q:=\sum_{i\in[h]} \ket{v_i} \bra{i}$.
Then the following are equivalent:
\begin{enumerate}
\item $\mathcal{H}$ is $M$-invariant;
\item $MQ=QX$ for an $h\times h$ matrix $X$;
\item $\mathcal{H}$ has an orthonormal basis consisting of eigenvectors of $M$.
\end{enumerate}
\label{lem:LinAlg}
\end{lemma}

\begin{proof}
Note that the operator $Q$ maps $\C^d$ to the subspace $\mathcal{H}$
and the action of $X$ on $C^d$ (in the basis $\ket{v_i}$)
represents the action of $M$ on $\mathcal{H}$ (in the basis
$\ket{v_i}$).

We first prove $(1)\Rightarrow(2)$. If $\mathcal{H}$ is $M$-invariant, then for all $i\in[h]$
\begin{equation}
  M \ket{v_i} = \sum_{j\in[h]} x_{ji} \ket{v_j}
\end{equation}
for some $x_{ji}\in\C$. If we let $X:=(x_{ij})$, then
\begin{equation}
  MQ = \sum_{i,j\in[h]} x_{ji} \ketbra{v_j}{i}
     = \sum_{j\in[h]} \Bigg( \ket{v_j}  \sum_{i\in[h]} \bra{i} x_{ji} \Bigg)
     = \sum_{j\in[h]} \ketbra{v_j}{j} X = QX.
\end{equation}

Next, let us show that $(2) \Rightarrow (3)$. If $MQ = QX$, then $X=Q\ct M Q$ as $Q\ct Q = \id_h$. Since $M$ is Hermitian, $X$ is also Hermitian and it has a spectral decomposition $X = \sum_{i \in [h]} \lambda_i \proj{w_i}$. Then for all $i\in[h]$
\begin{equation}
  M Q \ket{w_i} = Q X \ket{w_i} = \lambda_i Q \ket{w_i}.
\end{equation}
Therefore, the vectors $Q\ket{w_i}\in\mathcal{H}$ are eigenvectors of $M$. Finally, for all $i,j\in[h]$ we have $\bra{w_i}Q\ct Q\ket{w_j} = \braket{w_i}{w_j} = \delta_{ij}$. So the set
\begin{equation}
  \set{Q\ket{w_i}}_{i\in[h]}
\end{equation}
is an orthonormal basis of $\mathcal{H}$ consisting of eigenvectors of $M$.

Last, we prove that $(3)\Rightarrow (1)$. Let $\set{\ket{u_i}}_{i\in[h]} \subseteq \mathcal{H}$ be a set of orthogonal eigenvectors of $M$ with corresponding eigenvalues $\mu_i$. Then any vector $u\in\mathcal{H}$ can be expressed as $\ket{u} = \sum_{i\in[h]} c_i\ket{u_i}$ for some $c_i\in\C$. Now we have 
\begin{equation}
  M \ket{u} = \sum_{i\in[h]} \mu_i c_i\ket{u_i} \in \mathcal{H}
\end{equation}
as desired.
\end{proof}

We now show that whenever $a\otimes b \in \Pos(\bip{d_A}{d_B})$ is non-disturbing for a full orthogonal set of product states, so are $a \otimes \id$ and $\id \otimes b$.

\begin{lemma}
Let $a\otimes b\in \Pos(\bip{d_A}{d_B})$ and $P=P_A\otimes P_B\in \Pos(\bip{d_A}{d_B})$ be a projector onto a subspace $\mathcal{H}=\mathcal{H}_A\otimes\mathcal{H}_B$.
If $(a\otimes b) P\neq 0$ and $\mathcal{H}$ is $(a\otimes b)$-invariant, then $\mathcal{H}$ is also $(a\otimes \id)$ and $(\id\otimes b)$-invariant.
\label{lem:Nondist}
\end{lemma}

\begin{proof}
Let $h_A,h_B$ be the dimensions of $\mathcal{H}_A$ and $\mathcal{H}_B$,  respectively. Fix some orthonormal basis $\set{\ket{\alpha_i}}_{i\in[h_A]}$ of $\mathcal{H}_A$ and let $Q_A = \sum_{i\in[h_A]} \ket{\alpha_i} \bra{i}$, where $\ket{i}\in\C^{h_A}$. Define $Q_B$ similarly. Then $P_A = Q_A Q_A\ct$ and $P_B = Q_B Q_B\ct$. If $\mathcal{H}$ is $(a\otimes b)$-invariant, then 
\begin{equation}
  (a\otimes b)(Q_A\otimes Q_B) = (Q_A\otimes Q_B)X
\label{eq:Invariant}
\end{equation}
for an $(h_A h_B)\times (h_A h_B)$ matrix $X$. Note that $X$ is a tensor product, since $X = (Q_A\ct a Q_A)\otimes (Q_B\ct b Q_B)$. Since $(a\otimes b) P\neq0$ we also have that $(a\otimes b)(Q_A\otimes Q_B)\neq 0$. Hence, Equation~(\ref{eq:Invariant}) together with the fact that $X$ is a tensor product implies that
\begin{equation}
  a Q_A = Q_A X_A \text{ and } b Q_B = Q_B X_B 
\label{eq:Invariant2}
\end{equation}
for some $X_A$ and $X_B$ such that $X=X_A\otimes X_B$. By \Lem{LinAlg}, \Eq{Invariant2} implies that $\mathcal{H}_A$ is $a$-invariant and  $\mathcal{H}_B$ is $b$-invariant. Since any subspace is invariant under the identity operation, the lemma follows.
\end{proof}

We are ready to prove \Thm{States} using \Lem{LinAlg} and 
\Lem{Nondist}. 

\begin{proof}[Proof of Theorem~\ref{thm:States}]
We prove by induction on $d_A+d_B$. Clearly the states in $S$ can be discriminated with both finite and asymptotic LOCC if $d_A+d_B \leq 3$. We assume that the theorem statement holds for all values $d_A+d_B < m$ for some \mbox{$m\in \N$}. 

Suppose $d_A+d_B = m$ and the $n$ states in $S$ can be discriminated with asymptotic LOCC. Then for every $1/n\leq \chi \leq 1$ there exists a product operator $E = a \otimes b\in\Pos(\bip{d_A}{d_B})$ satisfying the three conditions in \Thm{KKB}. Our goal is to choose appropriate value of $\chi$ and use \Lem{Nondist} to conclude that both $a\otimes \id$ and $\id \otimes b$ are non-disturbing for $S$.

Pick any $\chi\in(\frac{1}{n},\frac{1}{n-1})$ and let $a\otimes b$ be the corresponding operator. Let us now check that $a\otimes b$ is nontrivial and satisfies the hypothesis of \Lem{Nondist}. The range of $\chi$ is chosen so that Conditions (1) and (2) together imply that
\begin{itemize}
\item $a\otimes b$ cannot be proportional to the identity matrix (from now on we assume that $a$ is not proportional to the identity matrix as the other case is similar);
\item for all $\rho\in S$ we have $E\rho\neq 0$.
\end{itemize}
For each $i\in[n]$, let $\mathcal{H}^{(i)} = \mathcal{H}^{(i)}_A \otimes \mathcal{H}^{(i)}_B$ be the column space of $\rho_i$ and $P_i$ be the projector onto $\mathcal{H}^{(i)}$. Then the last item implies that $(a\otimes b) P_i\neq 0$. Since $a\otimes b$ is non-disturbing for $S$, the subspace $\mathcal{H}^{(i)}$ is $(a\otimes b)$-invariant. Due to \Lem{Nondist}, $\mathcal{H}^{(i)}$ is $(a\otimes I)$-invariant.

Let $a_\lambda$ be the projector onto the $\lambda$-eigenspace of
$a$. Due to the equivalence of (1) and (3) in \Lem{LinAlg}, 
the subspaces $\mathcal{H}^{(i)}$ are $(a_\lambda\otimes
\id)$-invariant for all $\lambda\in\spec(a)$. So if $\mc{\id}_B$ is
the identity measurement on Bob, the nontrivial local projective
measurement
\begin{equation}
  \set{a_\lambda : \lambda\in\spec(a)} \otimes \mc{\id}_B 
  =:\mathcal{M}_A \otimes \mc{\id}_B
\end{equation}
is non-disturbing for $S$.

Suppose we measure the states in $S$ using $\mathcal{M}_A \otimes \mc{\id}_B$ and obtain outcome $\lambda\in\spec(a)$. If we restrict the unnormalized post-measurement states to the column space of $a_\lambda \otimes \id$, we have the following set:
\begin{equation}
  S_\lambda := \set{ (Q_\lambda \otimes \id)\ct \rho_i (Q_\lambda 
                         \otimes \id)}_{i\in[n]} = 
   \set[\Big]{\big(Q_\lambda\ct \tau_i Q_\lambda\big)
   \otimes \sigma_i}_{i\in[n]}
   \subseteq{\bip{\rank(a_\lambda)}{d_B}}.
\end{equation}
Here,
\begin{equation}
  Q_\lambda := \sum_{i\in[\rank(a_\lambda)]} \ketbra{\lambda_i}{i},
\end{equation}
$\ket{i}\in\C^{\rank(a_\lambda)}$, and $\set{\ket{\lambda_i}}_i \subseteq \C^{d_A}$ is some orthonormal basis of the $\lambda$-eigenspace of $a$. We now want to use the induction hypothesis to conclude that the states in $S_\lambda$ can be discriminated with finite LOCC. To do so, we have to check that $S_\lambda$ is a set of mutually orthogonal states that can be discriminated with asymptotic LOCC and that $\sum_{\rho\in S_\lambda} \rho$ has full rank.

First, since $\sum_{i\in[n]} \rho_i$ is positive semidefinite and has full rank and $Q_\lambda$ has full column rank, the matrix
\begin{equation}
  \sum_{\rho \in S_\lambda} \rho = 
   (Q_\lambda \otimes \id) \ct
   \Bigg( \sum_{i\in[n]} \rho_i \Bigg) (Q_\lambda \otimes \id)
\end{equation}
has full rank. Suppose that a sequence $\mathcal{R}_1,\mathcal{R}_2,\dotsc$ of finite LOCC protocols can be used to certify that the states in $S$ can be discriminated with asymptotic LOCC. Let $\mathcal{R}'_i$ be the finite LOCC protocol in which Alice first embeds her input space $\C^{\rank(a_\lambda)}$ in $\C^{d_A}$ by applying the isometry $Q_\lambda$ and then the two parties proceed with the protocol $\mathcal{R}_i$. After the embedding, Alice and Bob have the states $(a_\lambda \otimes \id) \rho_i (a_\lambda \otimes \id)$ up to a normalization. Since the column space, $\mathcal{H}^{(i)}$, of $\rho_i$ is $(a_\lambda \otimes \id)$-invariant, the column space of $(a_\lambda \otimes \id)$ is contained in $\mathcal{H}^{(i)}$.
Therefore, the sequence $\mathcal{R}'_1,\mathcal{R}'_2,\dotsc$ can be used to certify the asymptotic distinguishability of the states from $S_\lambda$. 

\begin{figure}[!ht]
\centering
\includegraphics[bb=0 -1 228 98,width=4in]{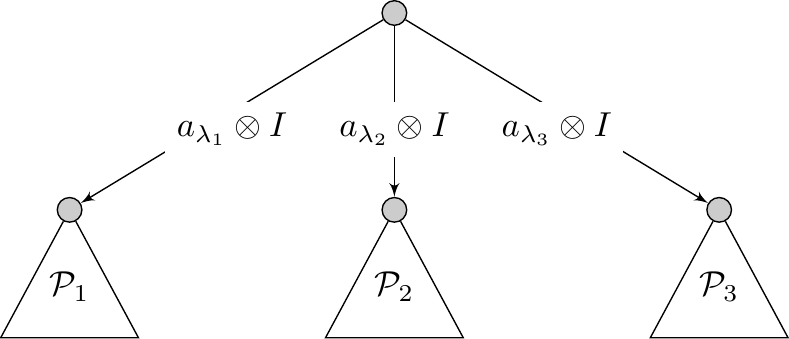}
\caption{An example, where $a$ has three distinct eigenvalues $\lambda_1, \lambda_2$ and $\lambda_3$. We first perform the measurement $\mathcal{M}_A\otimes \mc{\id}_B$ and then, conditioned on the outcome $\lambda_i$, proceed with the protocol $\mathcal{P}_i$ that discriminates the states from $S_{\lambda_i}$.}
\label{fig:ProductStates}
\end{figure}

\comment{
\begin{figure}[!ht]
\centering
\includegraphics[bb=0 -1 228 98]{Product-States.pdf}
\caption[An example of a measurement whose implementation can be initiated by a nontrivial measurement on Alice.]{An example, where $a$ has three distinct eigenvalues $\lambda_1, \lambda_2$ and $\lambda_3$. We first perform the measurement $\mathcal{M}_A\otimes \mc{\id}_B$ and then, conditioned on the outcome $\lambda_i$, proceed with the protocol $\mathcal{P}_i$ that discriminates the states from $S_{\lambda_i}$.}
\label{fig:ProductStates}
\end{figure}
}

Since $\mathcal{M}_A \otimes \mc{\id}_B$ is non-disturbing for $S$, the states in $S_\lambda$ are mutually orthogonal. Finally, as $\rank(a_\lambda)+d_B < d_A+ d_B$, the states from $S_\lambda$ can be discriminated by a finite LOCC protocol $\mathcal{P}$ by induction hypothesis. Combining the measurement $\mathcal{M}_A \otimes \mc{\id}_B$  with the finite LOCC protocol $\mathcal{P}$ gives a finite LOCC protocol for discriminating the states in~$S$ (See \Fig{ProductStates}).
\end{proof}

We can lift the tensor product requirement for one of the states in \Thm{States}.

\begin{cor}
Let $S= \set{\rho_i}_{i\in[n]}\subseteq\Pos(\bip{d_A}{d_B})$ be a full orthogonal set and assume that all but one $\rho_i$ can be expressed as $\rho_i = \sigma_i \otimes \tau_i$. If the states from $S$ can be discriminated with asymptotic LOCC then they can be discriminated with finite LOCC.
\label{cor:Tensor}
\end{cor}

\begin{proof}
Suppose that $\rho_1$ is the state that is not a tensor product. The proof is similar to that of \Thm{States}, except we cannot use \Lem{Nondist} to conclude that $\mathcal{H}^{(1)}$ is $(a\otimes \id)$ invariant. Instead we use the fact that the orthogonal subspaces $\mathcal{H}^{(2)},\dotsc,\mathcal{H}^{(n)}$ are all $(a\otimes\id)$-invariant, $\bigoplus_{i\in\set{2,\dotsc,n}}\mathcal{H}^{(i)}=\big(\mathcal{H}^{(1)}\big)^\perp$, and $a$ is Hermitian, to conclude that $\mathcal{H}^{(1)}$ is $(a\otimes\id)$-invariant.
\end{proof}

The main obstacle in generalizing \Thm{Measurements} to all separable projective measurements is the lack of an analogue of \Lem{Nondist} for separable projectors $P$. For example, consider 
\begin{equation}
  P := \big( \proj{0}\otimes \proj{1} \big) + 
       \big( \proj{1}\otimes \proj{0-1} \big) +
      \big( \proj{0-1}\otimes \proj{2} \big)
\end{equation}
and $a \otimes b := \proj{1}\otimes \proj{0-1}$, where $\ket{0-1}:=(\ket{0}-\ket{1})/\sqrt{2}$. Let $\mc{H}$ be the space onto which $P$ projects. Although $\mc{H}$ is $(a\otimes b)$-invariant, it is neither $(a\otimes \id)$- nor $(\id\otimes b)$-invariant, since $(a\otimes \id)\ket{0-1,2}\notin \mc{H}$ and $(\id\otimes b)\ket{0,1}\notin \mc{H}$.

Therefore, the general question of whether the set of POVM
measurements implementable by $\LOCCN$ is closed remains open, despite 
partial progress presented by \Thm{Measurements}.

\subsection{Applications}
Although \Cor{Tensor} is only a slight generalization of \Thm{Measurements}, it provides answers to natural questions. For example, let us consider the following orthonormal product basis is known as the \emph{domino basis} \cite{IBM} :
\[
\begin{array}{ccc}
\ket{\psi_0} = \ket{1} \ket{1},~~  
 & \ket{\psi_1^\pm} = \ket{0} \ket{0 \pm 1},~~ & \ket{\psi_{2}^\pm} = \ket{0 \pm 1}\ket{2}, \\
 & \ket{\psi_{3}^\pm} = \ket{2} \ket{1 \pm 2},~~ &  \ket{\psi_{4}^\pm} = \ket{1 \pm 2} \ket{0}, 
\end{array}
\]
where $\ket{i \pm j} := (\ket{i} \pm \ket{j})/\sqrt{2}$. It is known that the domino states cannot be discriminated by asymptotic LOCC \cite{IBM}. However, as soon as we modify the problem slightly the answer becomes unclear. For example, it is not known whether the states from
\begin{equation}
  S := \set[\Big]{ \proj{\psi_i^+}: i\in [4]} \cup 
      \set[\Big]{\rho:= \proj{\psi_0}+\frac{1}{4}\sum_{i\in[4]} \proj{\psi_i^-}},
\end{equation}
can be discriminated with asymptotic LOCC.  Questions about asymptotic
LOCC can be difficult to answer. In some settings, \Thm{Measurements}
and \Cor{Tensor} allows us to reduce such questions to those about
finite LOCC which usually are more tractable, as demonstrated in 
the following example. 

\begin{lemma}
Let $S=\set{\proj{\phi_i} : i\in [4]} \cup \set{\rho}$ be such that $\ket{\phi_i} = \ket{\psi_i^+}$ or $\ket{\phi_i} = \ket{\psi_i^-}$ and $\rho$ is the uniform mixture of the remaining $5$ domino states. Then the states from $S$ cannot be discriminated with $\ALOCC$.
\label{lem:Finite}
\end{lemma}
\begin{proof}
Because of \Cor{Tensor}, it suffices to prove that $S$ cannot be discriminated with \LOCCN.  
To do so, we only need to disprove the existence of nontrivial
(\ie, not proportional to the identity) positive semidefinite operators of the form $a\otimes\id$ and $\id \otimes b$. 

Let $\ket{\alpha_i}\ket{\beta_i} = \ket{\phi_i}$ and assume $a\otimes\id\in\Pos(\bip{3}{3})$ is non-disturbing for the set of states $S$. Then 
\begin{equation}
  0 = \bra{1,1} (a \otimes \id) \ket{\phi_1} = 
  \bra{1} a  \ket{0} \braket{1}{\beta_1} 
  \; \; \Leftrightarrow \; \;
  a_{10} = 0
  \; \; \Leftrightarrow \; \;
  a_{01} =0
\end{equation}
since $\braket{1}{\beta_1} \neq 0$ both for $\ket{\beta_1} = \ket{0+1}$ and $\ket{\beta_i} = \ket{0-1}$ and $a$ is Hermitian. Similarly, from $0 = \bra{1,1} (a \otimes \id) \ket{2,\beta_3}$ we obtain $a_{12}=a_{21}=0$. 
Either $\ket{\psi_2^+}$ or $\ket{\psi_2^-}$ belongs to the support of a different state from $S$ than $\ket{\psi_3^+}$. Hence, either $\bra{0-1}a\ket{2} = 0$ or $\bra{0+1}a\ket{2} = 0$. In both cases we obtain  $a_{02}=a_{20}=0$. To see that all the diagonal elements have to be equal, note that
$0 = \bra{0+1,2} (a \otimes \id) \ket{0-1,2} = a_{00}-a_{11}$ and $0 = \bra{1+2,0} (a \otimes \id) \ket{1-2,0} = a_{11}-a_{22}$. Thus, $a$ is proportional to the identity matrix. Via similar analysis, one can reach the same conclusion for $\id\otimes b$. Therefore, all the non-disturbing operators for $S$ are proportional to the identity matrix and our lemma follows.
\end{proof}

Building on \cite{KKB}, \cite{Chitambar-2013b} presents a
necessary condition for discriminating two states with asymptotic
LOCC.  When the support of these two states cover the whole
space this condition yields a simple criterion. This allows
the authors to show that
\begin{equation}
 \rho_+:= \frac{1}{4}\sum_{i\in[4]} \proj{\psi_i^+}
 \text{ and }
 \rho_-:= \frac{1}{5} \Big(\proj{1,1}+\sum_{i\in[4]} \proj{\psi_i^-} \Big)
\end{equation}
cannot be discriminated with asymptotic LOCC. Since our set $S$ is a
refinement of $\set[\big]{\rho_+,\rho_-}$ the impossibility to
distinguish the states from $S$ with asymptotic LOCC also follows from
the result of \cite{Chitambar-2013b}.  We illustrate
\Cor{Tensor} on the domino states because they are well-known.
There are cases where \Cor{Tensor} applies but the criterion 
from \cite{Chitambar-2013b} cannot be used to conclude that a set of
states cannot be discriminated with asymptotic LOCC.


\section{UBP's in $3\otimes3$ cannot be discriminated by $\ALOCC$}
\label{sec:UPB}

In this section, we turn our attention to the problem of
discrimination of unextendible product bases (UPBs).  It has been
shown in \cite{Terhal-1999,DiVincenzo-2003a} that UPBs cannot be perfectly discriminated
by finite LOCC.  Here, we show that any UPB in $\C^3 \otimes \C^3$
cannot be discriminated by asymptotic LOCC.

Reference \cite{DiVincenzo-2003a} establishes that any UPB in $\C^3 \otimes \C^3$ has
exactly $5$ states of the form $|\psi_i\> = |\a_i\> \otimes |\b_i\>$
which, up to local unitary transformations, can
be parametrized by six angles $\t_A$, $\g_A$, $\phi_A$, $\t_B$, $\g_B$, 
$\phi_B$: 
\bea
|\a_0\> & = & |0\>   \nonumber \\
|\a_1\> & = & |1\>   \nonumber \\
|\a_2\> & = & \cos \t_A |0\> + \sin \t_A |2\>  \nonumber \\
|\a_3\> & = & \sin \g_A \sin \t_A |0\> + \cos \g_{\!A} \, e^{i\phi_A} |1\> 
              - \sin \g_A \cos \t_A |2\>  \nonumber \\
|\a_4\> & = & \frac{1}{N_A} (\sin \g_A \cos \t_A \, e^{i\phi_A} |1\> 
             + \cos \g_A |2\>)  \nonumber \\
|\b_0\> & = & |1\>   \nonumber \\
|\b_1\> & = & \sin \g_B \sin \t_B |0\> + \cos \g_B \, e^{i\phi_B} |1\> 
              - \sin \g_B \cos \t_B |2\>  \nonumber \\
|\b_2\> & = & |0\>   \nonumber \\
|\b_3\> & = & \cos \t_B |0\> + \sin \t_B |2\>  \nonumber \\
|\b_4\> & = & \frac{1}{N_B} (\sin \g_B \cos \t_B \, e^{i\phi_B} |1\> 
             + \cos \g_B |2\>) 
\label{eq:upb}
\eea 
where $N_{A,B} = \sqrt{\cos^2 \g_{A,B} + \sin^2 \g_{A,B} \, \cos^2
  \t_{A,B}}$, and unextendibility implies that none of $\cos \g_{A,B}$,
$\sin \g_{A,B}$, $\cos \t_{A,B}$, $\sin \t_{A,B}$ vanishes.

\begin{theorem} \label{thm:upb-no-go}
The set of states $S = \{|\psi_i\>\}_{i=0,1,\cdots,4}$ cannot be 
discriminated by asymptotic LOCC.
\end{theorem}

\begin{proof} 
First, we show why it suffices to prove the case for $\phi_{A,B} = 0$.
We can replace $|1\>$ by $|\widetilde{1}\> = e^{i\phi_A} |1\>$ in the
choice of the local basis on Alice's system.  Then, $\phi_A$ only
appears in $|\a_1\> = e^{-i\phi_A} |\widetilde{1}\>$.  Apply a
similar change of basis on Bob's system.  Clearly, replacing $|\a_1\>$
and $|\b_0\>$ by $|\widetilde{1}\>$ does not affect
distinguishability.

We now proceed to prove the theorem by contradiction.  Suppose $S$ can
be discriminated by asymptotic LOCC.

Then, Theorem \ref{thm:KKB} applies to $S$ since the kernel condition
holds automatically when $S$ is a UPB.  Here, $n = 5$ and we take
$\chi = 0.22$ (any $1/n < \chi < 1/(n{-}1)$ will do).  The theorem
states that $\exists E=a\otimes b \geq 0$ that is non-disturbing for
$S$.  Furthermore, $E \neq 0$ (else $\sum_{|\psi\> \in S}\tr(E
|\psi\>\<\psi|) = 0 \neq 1$), $E \not \propto I$ (else, $\max_{|\psi\>
  \in S} \tr(E |\psi\>\<\psi|)=\tfrac{1}{5} \sum_{|\psi\> \in S}\tr(E
|\psi\>\<\psi|) = 1/5 \neq 0.22$), and $\min_{|\psi\> \in
  S} \tr(E |\psi\>\<\psi|) > 0.12$.  From these we derive constraints
for $a$ and $b$.  We have $a, b \neq 0$, $a$ cannot have two nonzero
eigenvalues of opposite signs (else the same holds for $E = a \otimes
b$) and similarly for $b$, so, without loss of generality, $a \geq 0$,
$b \geq 0$.  Finally, $\min_i \tr(a |\a_i\>\<\a_i|) > 0$ and $\min_i
\tr(b |\b_i\>\<\b_i|) > 0$.

We will show that the non-disturbing property of $E$ is inconsistent
with the conditions on $a,b$.

Our main tool in the analysis is an extension of the orthogonality
graph for UPBs defined in \cite{DiVincenzo-2003a}.  Given $a,b \in \Herm(\C^3)$,
define two graphs $G_a, G_b$ as follows.  They both have vertex set $V
= \{0,1,\dotsc,4\}$.  $G_a$ ($G_b$) has an edge $(i,j)$ whenever
$\<\a_i|a|\a_j\> = 0$ ($\<\b_i|b|\b_j\> = 0$).  Since $E$ is
non-disturbing for $S$, the pair $(i,j)$ is an edge in
$G_a$ or $G_b$ for all distinct $i,j$.  Since any three $|\a_i\>$ span $\C^3$, no vertex in
$G_a$ has degree more than $2$, and similarly for $G_b$.  So, the 
only possible $G_a, G_b$ are $5$-cycles with complementary sets of edges.

Denote the $12$ possible $5$-cycles as $O_{1,\cdots,12}$.  We analyse
the $12$ possible cases for which $G_a = O_{1,\cdots,12}$ ($G_b$ is
then fixed) in detail in the appendix.  We will see that in one case,
$E \propto I$ which is a contradiction.  For all other cases, 
$\min_i \tr(a |\a_i\>\<\a_i|) = 0$ which is also a contradiction.

\comment{
Here, we describe the methodology in the analysis.  First,
interchanging $|0\>$ and $|1\>$ corresponds to a permutation of the
vertices $0 \leftrightarrow 1$ and $2 \leftrightarrow 4$ (with
modified values of $\t_A, \g_A$), and there are only $8$
cases to study.  In each case, $a|\a_0\>$ lives 
in a one-dimensional subspace because it is orthogonal to
two (linearly independent) $|\a_i\>$'s where $i$ is adjacent to $0$ in
$G_a$.  Similarly for $a|\a_1\>$.  We thus obtain the first two
columns of $a$ in terms of $\t_A, \g_A$ and two scalars.  Hermiticity
of $a$ and the remaining orthogonality conditions further restrict the
form of $a$ until it is clear that either $a \propto I$ (in which case
we show $b \propto I$ and $E \propto I$) or $a \not \geq 0$, or
$\min_i \tr(a |\a_i\>\<\a_i|) = 0$.
}

\end{proof} 

We note as a side remark that, using Theorems 2 and 3 in
\cite{DiVincenzo-2003a}, any bipartite UPB with $5$ states can be
perfectly discriminated by a separable measurement.  \Thm{upb-no-go}
thus provides another example of the phenomenon {\em nonlocality
  without entanglement}, in which a set of unentangled states cannot
be discriminated by asymptotic LOCC, but can be discriminated by
separable operations.

\section{Discussions}
\label{sec:Conclusion}

To ease the analysis of asymptotic LOCC we have introduced two scenarios in which 
no new task can be accomplished (information theoretically) 
collapse.  by allowing vanishing error.  The first scenario is the
implementation of projective measurements with tensor product
operators. The second is the discrimination of the states from an
unextendible product basis in $\bip{3}{3}$.  On the first subject, an
obvious next step is to investigate whether asymptotic LOCC can be
helpful for implementing any projective or general measurement.  A
second question is whether asymptotic LOCC can help for perfect state
discrimination.  On the second subject, it is likely that a general
UPB cannot be discriminated by asymptotic LOCC.  It will be nice to
obtain a rigorous proof.

Another very poorly understood subject in LOCC is round
complexity. Almost nothing is known about how many messages the
parties need to exchange in order to accomplish a task in the LOCC
setting, especially when a small probability of error is allowed.


\section{Acknowledgements}
DL is supported by NSERC, NSERC DAS, CRC, and CIFAR.  LM is funded by the Ministry of Education (MOE) and National Research Foundation Singapore, as well as MOE Tier 3 Grant ``Random numbers from quantum processes'' (MOE2012-T3-1-009).

\clearpage 


\newcommand{\etalchar}[1]{$^{#1}$}

\clearpage

\appendix

\section{Case analysis for Theorem \ref{thm:upb-no-go}}

For concreteness, we denote the $12$ possible $5$-cycles as follows: 
\begin{figure}[http]
\setlength{\unitlength}{0.5mm}
\centering
\begin{picture}(70,60)

\put(10,30){\circle*{3}}
\put(50,30){\circle*{3}}
\put(30,45){\circle*{3}}
\put(20,10){\circle*{3}}
\put(40,10){\circle*{3}}

\put(30,45){\line(-4,-3){20}}
\put(10,30){\line(1,-2){10}}
\put(20,10){\line(1,0){20}}
\put(40,10){\line(1,2){10}}
\put(50,30){\line(-4,3){20}}

\put(02,48){\makebox(10,10){$O_1$:}}

\put(25,47){\makebox(10,10){$0$}}
\put(00,26){\makebox(10,10){$1$}}
\put(11,01){\makebox(10,10){$2$}}
\put(39,01){\makebox(10,10){$3$}}
\put(50,26){\makebox(10,10){$4$}}

\end{picture}
\begin{picture}(70,60)

\put(10,30){\circle*{3}}
\put(50,30){\circle*{3}}
\put(30,45){\circle*{3}}
\put(20,10){\circle*{3}}
\put(40,10){\circle*{3}}

\put(30,45){\line(-4,-3){20}}
\put(10,30){\line(1,-2){10}}
\put(20,10){\line(1,0){20}}
\put(40,10){\line(1,2){10}}
\put(50,30){\line(-4,3){20}}

\put(02,48){\makebox(10,10){$O_2$:}}

\put(25,47){\makebox(10,10){$0$}}
\put(00,26){\makebox(10,10){$1$}}
\put(11,01){\makebox(10,10){$2$}}
\put(39,01){\makebox(10,10){$4$}}
\put(50,26){\makebox(10,10){$3$}}

\end{picture}
\begin{picture}(70,60)

\put(10,30){\circle*{3}}
\put(50,30){\circle*{3}}
\put(30,45){\circle*{3}}
\put(20,10){\circle*{3}}
\put(40,10){\circle*{3}}

\put(30,45){\line(-4,-3){20}}
\put(10,30){\line(1,-2){10}}
\put(20,10){\line(1,0){20}}
\put(40,10){\line(1,2){10}}
\put(50,30){\line(-4,3){20}}

\put(02,48){\makebox(10,10){$O_3$:}}

\put(25,47){\makebox(10,10){$0$}}
\put(00,26){\makebox(10,10){$1$}}
\put(11,01){\makebox(10,10){$3$}}
\put(39,01){\makebox(10,10){$2$}}
\put(50,26){\makebox(10,10){$4$}}

\end{picture}
\begin{picture}(70,60)

\put(10,30){\circle*{3}}
\put(50,30){\circle*{3}}
\put(30,45){\circle*{3}}
\put(20,10){\circle*{3}}
\put(40,10){\circle*{3}}

\put(30,45){\line(-4,-3){20}}
\put(10,30){\line(1,-2){10}}
\put(20,10){\line(1,0){20}}
\put(40,10){\line(1,2){10}}
\put(50,30){\line(-4,3){20}}

\put(02,48){\makebox(10,10){$O_4$:}}

\put(25,47){\makebox(10,10){$0$}}
\put(00,26){\makebox(10,10){$1$}}
\put(11,01){\makebox(10,10){$3$}}
\put(39,01){\makebox(10,10){$4$}}
\put(50,26){\makebox(10,10){$2$}}

\end{picture}

\vspace*{2ex}

\begin{picture}(70,60)

\put(10,30){\circle*{3}}
\put(50,30){\circle*{3}}
\put(30,45){\circle*{3}}
\put(20,10){\circle*{3}}
\put(40,10){\circle*{3}}

\put(30,45){\line(-4,-3){20}}
\put(10,30){\line(1,-2){10}}
\put(20,10){\line(1,0){20}}
\put(40,10){\line(1,2){10}}
\put(50,30){\line(-4,3){20}}

\put(02,48){\makebox(10,10){$O_5$:}}

\put(25,47){\makebox(10,10){$0$}}
\put(00,26){\makebox(10,10){$1$}}
\put(11,01){\makebox(10,10){$4$}}
\put(39,01){\makebox(10,10){$2$}}
\put(50,26){\makebox(10,10){$3$}}

\end{picture}
\begin{picture}(70,60)

\put(10,30){\circle*{3}}
\put(50,30){\circle*{3}}
\put(30,45){\circle*{3}}
\put(20,10){\circle*{3}}
\put(40,10){\circle*{3}}

\put(30,45){\line(-4,-3){20}}
\put(10,30){\line(1,-2){10}}
\put(20,10){\line(1,0){20}}
\put(40,10){\line(1,2){10}}
\put(50,30){\line(-4,3){20}}

\put(02,48){\makebox(10,10){$O_6$:}}

\put(25,47){\makebox(10,10){$0$}}
\put(00,26){\makebox(10,10){$1$}}
\put(11,01){\makebox(10,10){$4$}}
\put(39,01){\makebox(10,10){$3$}}
\put(50,26){\makebox(10,10){$2$}}

\end{picture}
\begin{picture}(70,60)

\put(10,30){\circle*{3}}
\put(50,30){\circle*{3}}
\put(30,45){\circle*{3}}
\put(20,10){\circle*{3}}
\put(40,10){\circle*{3}}

\put(30,45){\line(-4,-3){20}}
\put(10,30){\line(1,-2){10}}
\put(20,10){\line(1,0){20}}
\put(40,10){\line(1,2){10}}
\put(50,30){\line(-4,3){20}}

\put(02,48){\makebox(10,10){$O_7$:}}

\put(25,47){\makebox(10,10){$0$}}
\put(00,26){\makebox(10,10){$2$}}
\put(11,01){\makebox(10,10){$4$}}
\put(39,01){\makebox(10,10){$1$}}
\put(50,26){\makebox(10,10){$3$}}

\end{picture}
\begin{picture}(70,60)

\put(10,30){\circle*{3}}
\put(50,30){\circle*{3}}
\put(30,45){\circle*{3}}
\put(20,10){\circle*{3}}
\put(40,10){\circle*{3}}

\put(30,45){\line(-4,-3){20}}
\put(10,30){\line(1,-2){10}}
\put(20,10){\line(1,0){20}}
\put(40,10){\line(1,2){10}}
\put(50,30){\line(-4,3){20}}

\put(02,48){\makebox(10,10){$O_8$:}}

\put(25,47){\makebox(10,10){$0$}}
\put(00,26){\makebox(10,10){$2$}}
\put(11,01){\makebox(10,10){$3$}}
\put(39,01){\makebox(10,10){$1$}}
\put(50,26){\makebox(10,10){$4$}}

\end{picture}

\vspace*{2ex}

\begin{picture}(70,60)

\put(10,30){\circle*{3}}
\put(50,30){\circle*{3}}
\put(30,45){\circle*{3}}
\put(20,10){\circle*{3}}
\put(40,10){\circle*{3}}

\put(30,45){\line(-4,-3){20}}
\put(10,30){\line(1,-2){10}}
\put(20,10){\line(1,0){20}}
\put(40,10){\line(1,2){10}}
\put(50,30){\line(-4,3){20}}

\put(02,48){\makebox(10,10){$O_9$:}}

\put(25,47){\makebox(10,10){$0$}}
\put(00,26){\makebox(10,10){$3$}}
\put(11,01){\makebox(10,10){$4$}}
\put(39,01){\makebox(10,10){$1$}}
\put(50,26){\makebox(10,10){$2$}}

\end{picture}
\begin{picture}(70,60)

\put(10,30){\circle*{3}}
\put(50,30){\circle*{3}}
\put(30,45){\circle*{3}}
\put(20,10){\circle*{3}}
\put(40,10){\circle*{3}}

\put(30,45){\line(-4,-3){20}}
\put(10,30){\line(1,-2){10}}
\put(20,10){\line(1,0){20}}
\put(40,10){\line(1,2){10}}
\put(50,30){\line(-4,3){20}}

\put(02,48){\makebox(10,10){$O_{10}$:}}

\put(25,47){\makebox(10,10){$0$}}
\put(00,26){\makebox(10,10){$3$}}
\put(11,01){\makebox(10,10){$2$}}
\put(39,01){\makebox(10,10){$1$}}
\put(50,26){\makebox(10,10){$4$}}

\end{picture}
\begin{picture}(70,60)

\put(10,30){\circle*{3}}
\put(50,30){\circle*{3}}
\put(30,45){\circle*{3}}
\put(20,10){\circle*{3}}
\put(40,10){\circle*{3}}

\put(30,45){\line(-4,-3){20}}
\put(10,30){\line(1,-2){10}}
\put(20,10){\line(1,0){20}}
\put(40,10){\line(1,2){10}}
\put(50,30){\line(-4,3){20}}

\put(02,48){\makebox(10,10){$O_{11}$:}}

\put(25,47){\makebox(10,10){$0$}}
\put(00,26){\makebox(10,10){$4$}}
\put(11,01){\makebox(10,10){$3$}}
\put(39,01){\makebox(10,10){$1$}}
\put(50,26){\makebox(10,10){$2$}}

\end{picture}
\begin{picture}(70,60)

\put(10,30){\circle*{3}}
\put(50,30){\circle*{3}}
\put(30,45){\circle*{3}}
\put(20,10){\circle*{3}}
\put(40,10){\circle*{3}}

\put(30,45){\line(-4,-3){20}}
\put(10,30){\line(1,-2){10}}
\put(20,10){\line(1,0){20}}
\put(40,10){\line(1,2){10}}
\put(50,30){\line(-4,3){20}}

\put(02,48){\makebox(10,10){$O_{12}$:}}

\put(25,47){\makebox(10,10){$0$}}
\put(00,26){\makebox(10,10){$4$}}
\put(11,01){\makebox(10,10){$2$}}
\put(39,01){\makebox(10,10){$1$}}
\put(50,26){\makebox(10,10){$3$}}

\end{picture}

\end{figure}

Omitting the subscript $A$ (which is irrelevant here), and using the
shorthands $c_\t, s_\t$ for $\cos \t$, $\sin \t$, $c_\g, s_\g$ for
$\cos \g$, $\sin \g$, we have
\bea
|\a_0\> & = & |0\>   \nonumber \\
|\a_1\> & = & |1\>   \nonumber \\
|\a_2\> & = & c_\t |0\> + s_\t |2\>  \nonumber \\
|\a_3\> & = & s_\g s_\t |0\> + c_\g |1\> - s_\g c_\t |2\>  
          =   s_\g |\a_2^{\perp}\> + c_\g |1\>
          =   s_\g s_\t |0\> + N |\a_4^{\perp}\>   \nonumber \\
|\a_4\> & = & \frac{1}{N} (s_\g c_\t |1\> + c_\g |2\>)  
\label{eq:upba}
\eea 
where $N = \sqrt{c_\g^2 + s_\g^2 c_\t^2}\,$, 
and we rephrase $|\a_3\>$ in terms of 
$|\a_2^{\perp}\> = s_\t|0\> - c_\t |2\>$, and 
$|\a_4^{\perp}\> = (c_\g |1\> - s_\g c_\t |2\>)/N$, two states 
that appear frequently in the analysis.  

If we swap $|0\>$ with $|1\>$, Eq.~(\ref{eq:upba}) becomes
\bea
|\a'_0\> & = & |1\>   \nonumber \\
|\a'_1\> & = & |0\>   \nonumber \\
|\a'_2\> & = & c_\t |1\> + s_\t |2\> 
           = \frac{1}{N'} (s_{\g'} c_{\t'} |1\> + c_{\g'} |2\>)
\nonumber \\
|\a'_3\> & = & s_\g s_\t |1\> + c_\g |0\> - s_\g c_\t |2\> 
           = c_{\g'} |1\> + s_{\g'} s_{\t'} |0\> - c_{\t'} s_{\g'}|2\>
\nonumber \\
|\a'_4\> & = & \frac{1}{N} (s_\g c_\t |0\> + c_\g |2\>) =  
               c_{\t'} |0\> +  s_{\t'} |2\>
\label{eq:upbanew}
\eea 
where $s_{\t'} = c_\g/N$, $c_{\t'} = s_\g c_\t/N$, $c_{\g'} = s_\g
s_\t$, $s_{\g'} = N$, and $N' = s_\g$.  So the local change of basis
$|0\> \leftrightarrow |1\>$ swaps $|\a_0\>$ with $|\a_1\>$ and swaps
$|\a_2\>$ with $|\a_4\>$ with modified angles.  Thus the analysis for
$G_a = O_{j}$ applies to the case $G_a = O_{j{+}1}$ for $j=2,4,8,10$.

Here, we summarize the methodology in the analysis.  In each case,
$a|\a_0\>$ lives in a one-dimensional subspace because it is
orthogonal to two (linearly independent) $|\a_i\>$'s where $i$ is
adjacent to $0$ in $G_a$.  Similarly for $a|\a_1\>$.  We thus obtain
the first two columns of $a$ in terms of $\t, \g$ and two scalar multipliers 
$r_{1,2}$ for $a|\a_{0,1}\>$.  We often use the original orthogonality
conditions between the $|\a_i\>$'s to deduce the form of
$a|\a_{0,1}\>$.  With the first two columns of $a$ fixed, we use
hermiticity of $a$ to fix all but one entry of $a$ (which we call
$r_3$).  We use the remaining orthogonality conditions to relate
$r_{1,2,3}$ until we either obtain $a \propto I$ (in which case we
show $b \propto I$ as well and thus $E \propto I$) or use $a \geq 0$
to force $r_1=0$ or $r_2=0$ thereby contradicting $\min_i \tr(a
|\a_i\>\<\a_i|) > 0$.

{\bf Case I:} If $G_a = O_1$, then \\
$\bullet$ $a|\a_0\> \perp |\a_1\>, |\a_4\>$, so, $a|\a_0\> = r_1 |\a_0\>$. \\
$\bullet$ $a|\a_1\> \perp |\a_0\>, |\a_2\>$, so, $a|\a_1\> = r_2 |\a_1\>$. \\
By hermiticity, 
$a = \left( \begin{array}{ccc} r_1 & 0 & 0 \\ 0 & r_2 & 0 \\ 0 & 0 & r_3 
            \end{array} \right)$.  
So, $a |\a_3\> = r_1 s_\g s_\t |0\> + r_2 c_\g |1\> - r_3 s_\g c_\t |2\>$. \\
Imposing $0 = \<a_2| a |\a_3\> = c_\t r_1 s_\g s_\t - s_\t r_3 s_\g c_\t$ 
gives $r_1 = r_3$, and 
$0 = N \<a_4| a |\a_3\> = s_\g c_\t r_2 c_\g - c_g r_3 s_\g c_\t$ gives
$r_2 = r_3$, so $a \propto I$.  

When $G_a = O_1$, $G_b = O_7$.  Up to different choices for
the angles, the states $|\b_{0,3,1,4,2}\>$ are the same as 
$|\a_{1,2,3,4,0}\>$, so, the above analysis applies to $G_b = O_7$ and 
$b \propto I$.  Thus $E = a \otimes b \propto I$ a contradiction.

{\bf Case II:} If $G_a = O_2$, then \\
$\bullet$ $a|\a_0\> \perp |\a_1\>, |\a_3\>$, so, $a|\a_0\> = r_1 |\a_2\>$. \\
$\bullet$ $a|\a_1\> \perp |\a_0\>, |\a_2\>$, so, $a|\a_1\> = r_2 |\a_1\>$. \\
By hermiticity, 
$a = \left( \begin{array}{ccc} r_1 c_\t & 0 & r_1 s_\t \\ 
                               0 & r_2 & 0 \\ r_1 s_\t & 0 & r_3 
            \end{array} \right)$.  
So, $a (N |\a_4\>) 
     = r_1 s_\t c_\g |0\> + r_2 s_\g c_\t |1\> + r_3 c_\g |2\>$. \\
Imposing $0 = N \<a_2| a |\a_4\> = c_\t r_1 s_\t c_\g + s_\t r_3 c_\g$ gives  
$c_\t r_1 = -r_3$.  But $a \geq 0$ implies non-negativity of all diagonal 
elements, so, $r_1 = 0$ and thus $\tr(a |\a_0\>\<\a_0|) = 0$ a 
contradiction.

The same analysis applies to $G_a = O_3$.

{\bf Case III:} If $G_a = O_4$, then \\
$\bullet$ $a|\a_0\> \perp |\a_1\>, |\a_2\>$, so, 
          $a|\a_0\> = r_1 |\a_2^{\perp}\>$. \\
$\bullet$ $a|\a_1\> \perp |\a_0\>, |\a_3\>$, so, 
          $a|\a_1\> = r_2 (N|\a_4\>)$.\\
By hermiticity, 
$a = \left( \begin{array}{ccc} r_1 s_\t & 0 & - r_1 c_\t \\ 
             0 & r_2 s_\g c_\t & r_2 c_\g \\ - r_1 c_\t & r_2 c_\g & r_3 
            \end{array} \right)$.  \\
So, $a (N |\a_4\>) 
     = - r_1 c_\t c_\g |0\> + r_2 (s_\g^2 c_\t^2 + c_\g^2) |1\> 
       + (r_2 c_\g s_\g c_\t + r_3 c_\g) |2\>$. \\
Now $0 = N \<a_2| a |\a_4\> = c_\g \left(
          -r_1 c_\t^2 + s_\t (r_2 s_\g c_\t + r_3 )  \right)$ so 
$s_\t r_3 = r_1 c_\t^2 - s_\t r_2 s_\g c_\t$.

Since $a \geq 0$, the minors $r_1 s_\t r_2 s_\g c_\t \geq 0$, 
$r_1 s_\t r_3 - r_1^2 c_\t^2 \geq 0$.  \\ Eliminating $r_3$ in 
the latter, $r_1 (r_1 c_\t^2 - s_\t r_2 s_\g c_\t) - r_1^2 c_\t^2 
= -r_1 s_\t r_2 s_\g c_\t \geq 0$.  So $r_1 r_2 = 0$, which means 
$\tr(a |\a_i\>\<\a_i|) = 0$ either for $i=0$ or $1$ which is 
a contradiction.
\\[0.7ex]
The same analysis applies to $G_a = O_5$.  


{\bf Case IV:} If $G_a = O_6$, then \\
$\bullet$ $a|\a_0\> \perp |\a_1\>, |\a_2\>$, so, 
          $a|\a_0\> = r_1 |\a_2^{\perp}\>$. \\
$\bullet$ $a|\a_1\> \perp |\a_0\>, |\a_4\>$, so, 
          $a|\a_1\> = r_2 (N|\a_4^{\perp}\>)$.\\
By hermiticity, 
$a = \left( \begin{array}{ccc} r_1 s_\t & 0 & - r_1 c_\t \\ 
       0 & r_2 c_\g & - r_2 s_\g c_\t \\ - r_1 c_\t & - r_2 s_\g c_\t & r_3 
            \end{array} \right)$.  \\
So, $a (|\a_3\>) 
     =  r_1 s_\g |0\> + r_2 (s_\g^2 c_\t^2 + c_\g^2) |1\> 
       - c_\t s_\g (r_1 s_\t + r_2 c_\g + r_3) |2\>$. \\
Now $0 = \<a_4| a |\a_3\> = 
     s_\g c_\t r_2 (s_\g^2 c_\t^2 + c_\g^2)
   - c_\g c_\t s_\g (r_1 s_\t + r_2 c_\g + r_3)
=    s_\g c_\t (r_2 s_\g^2 c_\t^2 - c_\g r_1 s_\t - c_\g r_3)$. 
So, $c_\g r_3 = r_2 s_\g^2 c_\t^2 - c_\g r_1 s_\t$.  

Since $a \geq 0$, the minors $r_1 s_\t r_2 c_\g \geq 0$, 
$r_2 c_\g r_3 - r_2^2 s_\g^2 c_\t^2 \geq 0$.  
\\ Eliminating $r_3$ in 
the latter, 
$r_2 (r_2 s_\g^2 c_\t^2 - c_\g r_1 s_\t) - r_2^2 s_\g^2 c_\t^2 \geq 0$.  
Simplifying, $r_2 (- c_\g r_1 s_\t) \geq 0$.  So, $r_1 r_2 = 0$, which 
is a contradiction (see case III).

{\bf Case V:} If $G_a = O_7$, then \\
$\bullet$ $a|\a_0\> \perp |\a_2\>, |\a_3\>$, so, 
          $a|\a_0\> = r_1 (s_\g |1\> - c_\g |\a_2^{\perp}\>)
                    = r_1 (-c_\g s_\t |0\> + s_\g |1\> + c_\g c_\t |2\>)$. \\
$\bullet$ $a|\a_1\> \perp |\a_3\>, |\a_4\>$, so, 
          $a|\a_1\> = r_2 N (N |0\> - s_\g s_\t |\a_4^{\perp}\>)
          = r_2 \left( (c_\g^2 + s_\g^2 c_\t^2) |0\> - s_\g s_\t c_\g |1\> 
                  + s_\g^2 s_\t c_\t |2\> \right) $.\\
Thus
$a = \left( \begin{array}{ccc} 
- r_1 c_\g s_\t & r_2 (c_\g^2 + s_\g^2 c_\t^2) & * \\ 
r_1 s_\g & - r_2 s_\g s_\t c_\g & *  \\ 
r_1 c_\g c_\t & r_2  s_\g^2 s_\t c_\t & *  
            \end{array} \right)$ where we omit the third column which 
does not enter the analysis.

By hermiticity, 
$r_1 s_\g =  r_2 (c_\g^2 + s_\g^2 c_\t^2)$.  
Multiplying both sides by $r_2$, we get $r_1 s_\g r_2 \geq 0$. 

The minor after deleting the second row and the third column is 
$r_1 r_2 s_\g (c_\g^2 s_\t^2  - c_\g^2 - s_\g^2 c_\t^2) \geq 0$.
The expression in the parenthesis is negative, so, $r_1 r_2 s_\g \leq 0$.  

Together, $r_1 r_2 = 0$, but that gives a contradition.

{\bf Case VI:} If $G_a = O_8$, then \\
$\bullet$ $a|\a_0\> \perp |\a_2\>, |\a_4\>$, so, 
          $a|\a_0\> = r_1 |\a_3\> 
                    = r_1 (s_\g s_\t |0\> + c_\g |1\> - s_\g c_\t |2\>)$. \\
$\bullet$ $a|\a_1\> \perp |\a_3\>, |\a_4\>$ (which is same as in case V).

Thus
$a = \left( \begin{array}{ccc} 
r_1 s_\g s_\t & r_2 (c_\g^2 + s_\g^2 c_\t^2) & * \\ 
r_1 c_\g & - r_2 s_\g s_\t c_\g & *  \\ 
- r_1 s_\g c_\t & r_2  s_\g^2 s_\t c_\t & *  
            \end{array} \right)$ where we omit the third column which 
does not enter the analysis.

By hermiticity, 
$r_1 c_\g =  r_2 (c_\g^2 + s_\g^2 c_\t^2)$, so,  $r_1 r_2 c_\g \geq 0$.  

The determinent of the $|0\>,|1\>$ block is 
$- r_1 r_2 c_\g (s_\g^2 s_\t^2 + s_\g^2 c_\t^2 + c_\g^2) =  - r_1 r_2 c_\g
\geq 0$.  

Together, $r_1 r_2 = 0$, giving a contradition. 

The same analysis applies to $G_a = O_9$.

{\bf Case VII:} If $G_a = O_{10}$, then $a|\a_0\> \perp |\a_3\>, |\a_4\>$ 
and $a|\a_1\> \perp |\a_2\>, |\a_4\>$. 

So, part of $a$ can be obtained from 
that in case VI with the first two columns interchanged:

$a = \left( \begin{array}{ccc} 
r_2 (c_\g^2 + s_\g^2 c_\t^2) & r_1 s_\g s_\t & * \\ 
- r_2 s_\g s_\t c_\g & r_1 c_\g & *  \\ 
r_2  s_\g^2 s_\t c_\t & - r_1 s_\g c_\t & *  
            \end{array} \right)$.

Hermiticity now implies $-r_2 c_\g = r_1$

The determinent of the $|0\>,|1\>$ block is minus that in case VI. 
So, $r_1 r_2 c_\g \geq 0$.  

Together, $-r_1^2 \geq 0$, so, $r_1 = 0$ giving a contradition. 

The same analysis applies to $G_a = O_{11}$.  

\clearpage 

{\bf Case VIII:} If $G_a = O_{12}$, then $a|\a_0\> \perp |\a_3\>, |\a_4\>$ 
and $a|\a_1\> \perp |\a_2\>, |\a_3\>$. 

So, part of $a$ can be obtained from that in case V with the first two
columns interchanged:

$a = \left( \begin{array}{ccc} 
r_2 (c_\g^2 + s_\g^2 c_\t^2) & - r_1 c_\g s_\t & r_2  s_\g^2 s_\t c_\t \\ 
- r_2 s_\g s_\t c_\g & r_1 s_\g & r_1 c_\g c_\t  \\ 
r_2  s_\g^2 s_\t c_\t & r_1 c_\g c_\t & r_3
            \end{array} \right)$ 

where we fill in part of the third column using hermiticity.  
Hermiticity also implies $r_2 s_\g = r_1$.

Since $a|\a_2\> \perp |\a_1\>, |\a_4\>$, so $a|\a_2\> \propto |0\>$ 
and \\
$0 = \<2|a|\a_2\> = \<2| a (c_\t |0\> + s_\t |2\>) 
       = c_\t \<2|a|0\> + s_\t \<2|a|2\> = c_\t r_2  s_\g^2 s_\t c_\t  
       + s_\t r_3$, so, $r_3 = - r_2  s_\g^2 c_\t^2$.  

Using the relation between $r_{1,2}$, we get $r_3 = -r_1 s_\g c_\t^2$.

Since $a \geq 0$, product of the last two diagonal elements is non-negative.
So,  $r_1 s_\g (-r_1 s_\g c_\t^2) \geq 0$, and $r_1 = 0$ a contradiction.

\end{document}